\documentclass[10pt]{article}
\usepackage{dsfont} 
\usepackage{graphicx}
\usepackage{graphics} 
\usepackage{float}

\usepackage{natbib}
\usepackage{apalike}
\usepackage{url}
\usepackage{amsfonts}
\usepackage{amsmath, amssymb}

\usepackage{here} 
\usepackage{multirow}
\usepackage{booktabs} 
\usepackage{dsfont }
\usepackage{color}
\usepackage{mathrsfs}

\usepackage[all]{xy}
    \usepackage{amsthm}
\newtheorem{lem}{Lemma}

\newtheorem{thm}{Theorem}
\newtheorem{prop}{Proposition}
\newcommand\ci{\perp\!\!\!\perp}

\usepackage{tikz}
\usetikzlibrary{snakes}
\usepackage{fancyhdr}
\usepackage[left=3cm,right=3cm,top=1.5cm,bottom=1.5cm,includeheadfoot]{geometry}

\title{Dynamic path analysis - A useful tool to investigate mediation processes in clinical survival trials}
\author{Susanne Strohmaier, Kjetil R{{\o}}ysland, Rune Hoff, {{\O}}rnulf Borgan, Terje R  Pedersen,\\ Odd O. Aalen }

\begin{document}
\maketitle

\begin{abstract}
When it comes to clinical survival trials, regulatory restrictions usually require the application of methods that solely utilize baseline covariates and the intention-to-treat principle. Thereby a lot of potentially useful information is lost, as collection of time-to-event data often goes hand in hand with collection of information on biomarkers and other internal time-dependent
covariates. 
However, there are tools to incorporate information from repeated measurements in a useful manner that can help to shed more light on the underlying treatment mechanisms. We consider dynamic path analysis, a model for mediation analysis in the presence of a time-to-event outcome and time-dependent covariates to investigate direct and indirect effects in a study of different lipid lowering treatments in patients with previous myocardial infarctions. 
Further, we address the question whether survival in itself may produce associations between the treatment and the mediator in dynamic path analysis and give an argument that, due to linearity of the assumed additive hazard model, this is not the case.
We further elaborate on our view that, when studying mediation, we are actually dealing with underlying processes rather than single variables measured only once during the study period. This becomes apparent in results from various models applied to the study of lipid lowering treatments as well as our additionally conducted simulation study, where we clearly observe, that discarding information on repeated measurements can lead to potentially erroneous conclusions.

\end{abstract}

\section{Introduction}
Survival and event history analysis has become a central interest in clinical biostatistics, as many clinical trials are designed to investigate the effect of certain treatments on the time until the occurrence of a particular event of interest, such as death or disease progression.  As required by regulatory authorities, typical approaches to assess these effects comprise of Kaplan-Meier plots along with log-rank tests or of employing a Cox regression model including treatment and baseline covariates. Those analysing strategies mainly focus on answering the rather pragmatic question 'Does treatment work?'. \\
Often collection of time-to-event data goes hand in hand with the collection of information on biomarkers and other internal time-dependent covariates, which are hardly ever used in the final analysis. Thereby a lot of useful information is ignored, that could be used to address the more exploratory question 'How does treatment work?'. A key tool to approaching answers to that question is mediation analysis, that allows for a decomposition of the total treatment effect into a direct effect and an indirect effect. Emsley and co authors \cite{emsley} give a historical overview of approaches to mediation analysis and thereby discuss two - at first glance distinct appearing - traditions. On the one hand, the estimation of direct and indirect effects by the method of path analysis together with structural equation modeling \cite{wright_path_analysis, duncan_path_analysis, goldberg_structural_equations, baron_kenny}  mainly motivated by social and behavioural sciences. On the other hand, the 'causal inference approach', mainly developed by statisticians and econometricians focusing on assumptions needed for the identification of direct and indirect effects to draw valid inference \cite{rubin_1974, holland_path_analysis, robins_greenland_identifiability, hafeman, cole_hernan}, also pointing out the problems that would occur when path-tracing rules and estimation techniques from linear models are transferred to non-linear models \cite{kaufman,pearl_mediation_formula}.\\
Besides other obvious connections between these two traditions, one common issue was, that until the last decade there was no straightforward approach on how to handle the estimation of direct and indirect effects for survival outcomes.
In 2006, Fosen and co-authors \cite{fosen,fosen2} proposed the model of dynamic path analysis,  which we will mainly be concerned with in this paper, based on linear regression and the additive hazard model \cite{aalen_additive80, aalen_additive} . The approach was originally motivated by recurrent event modeling and would - following the cartegorisation by Emsley and co-authors - fall into the tradition of path analysis. Lange and Hansen \cite{lange} developed a method from a causal inference perspective and presented a way to obtain natural direct and indirect effects. They make use of the additive hazard model as well, however, their approach is restricted to the setting of a time-fixed, normally distributed mediator.\\
The model of dynamic path analysis enables us to exploit more information from data routinely collected within clinical trials (or observational studies, as illustrated in \cite{roysland, gamborg}) by utilizing repeated measurements of the mediator. The model is based on the idea that we are actually dealing with continuous processes that evolve over time rather than fixed variables, as pointed out in \cite{believe}, and aims on modeling the effects of several covariate processes on the occurrence of the event of interest and the relation between the covariate processes. 
It can be viewed as an extension of classical path analysis \cite{wright_path_analysis} and the concept of directed acyclic graphs (DAGs), an important tool in causal inference, to settings that involve time-to-event outcomes and time-dependent covariates. The extensions are essentially, that a DAG is defined at each event time and forms therefore a stochastic process in itself, that the path coefficients may change over time and that the outcome is the occurrence of an event. However, making use of the additive hazard model for the time-to-event outcome and linear regression for the treatment-mediator relationship, preserves the rules for multiplying coefficients along paths, which would not be meaningful in non-linear models, like for example the Cox model \cite{kaufman,vanderweel_survival}. \\
Martinussen \cite{martinussen} derived the large sample properties of the dynamic path analysis model and he further states that under the additional assumptions that the treatment-outcome relation and the mediator-outcome relation are un-confounded and there are no interactions between the treatment and the mediator, the obtained estimates could be interpreted as truly causal effects. He, however, emphasises the notion that no unique definitions of the concepts of direct, indirect and total effects exist \cite{pearl_direct_indirect_effects, robins_greenland_identifiability,goetgeluk}.\\
Further considerations within the field of causal inference motivated by the model have mainly addressed various scenarios where measured or unmeasured confounders could occur on the pathway between the mediator and the outcome, which could themselves be affected by the exposure. For example, \cite{martinussen2} propose a two-stage estimator for the controlled direct effect of a point exposure on a survival outcome under one particular confounding situation. However, that approach is limited to a time-fixed intermediate variable as well.  \\
In this paper, we will assume that the treatment-outcome as well as the mediator-outcome relationships are unconfounded and that no treatment-mediator interactions are present, and rather address the question whether survival in itself may produce associations between the treatment and the mediator in dynamic path analysis. We give an argument showing that, due to linearity of the additive hazard model, this is not the case, so selection by survival does not produce artificial association. This implies the important property that, in case one finds the treatment mediator effect to be changing over time, this is due to a real change of the effect and not caused by the survival selection mechanism. A further important side-product of our argument is that covariates that affect survival in an additive manner but are not considered as confounders will not turn into a confounders as time passes by.\\
For ease of presentation we present the underlying derivations for a single time point in the main text. However, we provide an appendix that contains a generalisation of that argument, but involves a mathematically more precise and therefore perhaps more complicated appearing notation.\\
To elaborate further on our view, that we are actually dealing with processes, we performed a simulation study to compare the application of dynamic path analysis in the situation where the mediator is only measured at one point in time compared to utilizing several measurement in a setting where we actually assume the mediator to be a time-varying process. We can clearly observe that discarding information tends to obscure potential inference about the underlying mechanisms. 
Furthermore, we have data available that were collected within the IDEAL study project (\textbf{I}ncremental \textbf{D}ecrease in \textbf{E}nd Points Through \textbf{A}ggressive \textbf{L}ipid Lowering, \cite{ideal_main}), a multi-center clinical trial designed to compare the effects of two different lipid lowering strategies on the risk of cardiovascular disease among patients with established coronary heart disease (CHD). Enrolled patients were asked to repeatedly return to the study center at pre-scheduled time points for monitoring their lipid values along with other laboratory measurements. 
The trial was designed based on the mechanistic understanding that statins would lower low density lipoprotein (LDL)-cholesterol levels, what would, in turn, result in a reduced risk of coronary events. However, the trial did not show the expected results regarding the primary outcome. Among other more in-depth considerations, the researchers also investigated the proportion of the treatment effect mediated through different lipid measures \cite{ideal_meta_treatment_effect_explained}, but the applied methods were based on utilizing only a single measurement at a particular time point \cite{ideal_simes_landmarking, freedman_landmarking}. We employ dynamic path analysis to illustrate how the direct and the indirect effects develop over time and also use a more broadly defined outcome, to utilize more of the collected information. \\
The paper is structured as follows. First, we describe the concept of dynamic path analysis and the respective estimation equations, followed by our argument concerning conditioning on survival in Section \ref{survival_collider}. In Section \ref{simulation_study} we present results from our simulation study, followed by the results from application to the IDEAL data in Section \ref{IDEAL_analysis}, and a discussion in Section \ref{discussion}. In an appendix we present the generalisation of the argument explained in Section \ref{survival_collider} to various time points and motivate the causal interpretation of the described direct and indirect effects from Section \ref{dynamic_path}.

\section{Dynamic path analysis}\label{dynamic_path}
As mentioned above the idea behind dynamic path analysis is to define a series of graphical models, that depict the relationship between a time-fixed treatment (more baseline covariates could be added, without destroying the big picture), a covariate process and the occurrence of the event of interest, modelled as the infinitesimal change of a counting process. More specifically, we consider the situation illustrated in Figure \ref{DAG}, where $X_1$ would represent the statin treatment at randomisation in our application, $X_2(t-)$ stands for the time-dependent mediator, e.g. LDL-cholesterol, 
 and $\mathrm{d}N(t)$ essentially refers to a jump in the counting process, which would then be a 'any CHD event'. \\
\begin{figure}[H]
\begin{displaymath}
\xymatrix{
       & X_2(t-)  \ar[dr] ^{\beta_{3,2}(t)dt}&\\
X_1 \ar[ur] ^{\beta(t)_{2,1}} \ar@/_/[rr]_{\beta_{3,1}(t)dt} &              & dN(t)
}
\end{displaymath}
\caption{Illustration of one time-local DAG.}
\label{DAG}
\end{figure}
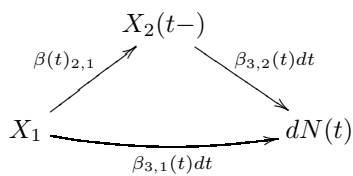 
To perform a dynamic path analysis one has to regress each node in the diagram onto its parents for each event time. For the mediator-exposure relationship this can be done by ordinary linear regression, whereas for the counting process increment we make use of the additive hazard model \cite{aalen_additive80,aalen_additive}. \\
Let us be more specific. Let $N(t) =\{N(t);t \in [0,\tau]\}$ denote a right continuous counting process for one particular individual. Given that we are considering survival data, N(t) will start at 0 and jump to 1 if the individual experiences the event of interest. We want to model the 'jump-intensity' within a small time interval, given all information that has been observed prior to that interval, which we heuristically denote as 'past'. Let $\mathrm{d}N(t)$ denote the increment of the counting process $N(t)$ over a small interval $[t, t + \mathrm{d}t)$, then we can express the intensity process $\alpha(t)$ 
 as 
$$
\alpha(t)\mathrm{d}t=P(\mathrm{d}N(t)=1|\text{ past}),
$$
which we can alternatively express as $E [\mathrm{d}N(t)|\text{ past}]$, since $ \mathrm{d}N(t)$ is a binary variable. 
Applying well-known results from counting process theory \cite{andersen, aalen_book}, we have the following decomposition
$$ 
\mathrm{d}N(t)= \lambda(t)\mathrm{d}t + \mathrm{d}M(t),
$$
where $\mathrm{d}M(t)$ is a martingale increment. Note that the decomposition is of the form 'observation = signal + noise' with $\mathrm{d}N(t)$ interpreted as the observation, $\lambda(t)\mathrm{d}t$ as the signal and $\mathrm{d}M(t)$ as the noise. We further assume independent censoring (\cite{aalen_book}, Chapter 2) and introduce an 'at risk indicator' $Y(t)$ at time $t$.\\
As stated above, for dynamic path analysis we want to make use of the additive hazard model to describe the relationship between the change in the counting process and the covariate processes,
which takes following form 
\begin{align}
\mathrm{d}N(t)= Y(t)\{\beta_{3,0}(t) + \beta_{3,1}(t)X_1 +\beta_{3,2}(t)X_2(t-) + \beta_{3,3}(t)Z_1 + \dots +\beta_{3,p+2}(t)Z_p\}\mathrm{d}t + \mathrm{d}M(t). \label{add_model}
\end{align}
Here the first term on the right-hand side describes the intensity process as the product of the 'at risk indicator', $Y(t)$, and a hazard function, $\alpha(t)$, corresponding to the term in curly brackets. Further, the $\beta_{3,j}(t)$ are arbitrary regression functions, $X_1$ denotes the time-fixed treatment, $X_2(t-)$ the time dependent mediator and the vector $\boldsymbol{Z}$ relevant baseline covariates. 
For the treatment-mediator relationship we assume a linear relationship
\begin{align} 
X_2(t)=b_{2,0}(t) + b_{2,1}(t)X_1+ b_{2,2}(t)Z_1 + \dots +b_{2,p+1}(t) Z_{p} + W_2(t),  \label{mediator}
\end{align}
where in this case $b_{2,j}(t)$ denote the respective regression coefficients at time $t$ and $W_2(t)$ the error term at time $t$.
In applications, of course, different baseline covariates can be incorporated in (\ref{add_model}) and (\ref{mediator}) depending on the setting, but we will use the index $p$ throughout for the sake of simplicity.
Further we also assume
\begin{align*}
E[X_2(t)|X_1,\boldsymbol{Z},T>t]=b_{2,0}(t) + b_{2,1}(t)X_1+ b_{2,2}(t)Z_1 + \dots +  b_{2,p+1}(t) Z_{p},
\end{align*}
what will be further discussed in Section \ref{survival_collider} and the appendix.
For consistency with later sections, we assume the following structure for $X_1$ 
\begin{align}\label{eq_x1}
X_1=b_{1,0} + W_1,
\end{align}
with the error term $W_1$ being independent of the error terms $W_2(t)$ at all times.
Note that the case of a binary treatment, as e.g. in our application section, is also covered in that formulation.
\\
Use of the additive hazards model is a key ingredient in our approach. As described in section 2.1, estimation in the additive hazards model focuses on the cumulative regression functions; see also section 4.2 in the book by Aalen et al. \cite{aalen_book}. For this reason we will (as in \cite{roysland}) define the cumulative indirect and direct effects from our structural equation models in equations (\ref{add_model} - \ref{eq_x1}) as
\begin{align}\label{dir}
\text{dir}(X_1\rightarrow N(t)): 
\int_0^t \beta_{3,1}(s) \mathrm{d}s
\end{align}
and 
\begin{align}\label{indir}
\text{indir}(X_1\rightarrow X_2(t) \rightarrow N(t)): \int_0^t{b_{2,1}(s)\beta_{3,2}(s)}\mathrm{d}s.
\end{align}
It should be noted, however, that our interest is in the local indirect and direct effects. So when interpreting the cumulative effect estimates, we focus on their slopes.
For more detailed elaborations see \cite{fosen}, where they also showed that the total cumulative effect can be obtained by simply taking the sum over the cumulative direct and indirect effects. \\
In our appendix we motivate a causal interpretation of the time local direct and indirect effects, integration then gives the respective cumulative effects. Heuristically speaking, the indirect effect reflects what would happen if we could manipulate the mediator $X_2(t)$ such that it behaved as if $X_1$ had been perturbed, while we let $X_1$ remain unperturbed. This idea resembles the definition of a natural indirect effect, but without making explicit use of the counterfactual framework. The idea underlying the direct effect is an intervention where $X_1$ is perturbed, but the effect that perturbation would have had on $X_2(t)$ can be attenuated. For a more general and mathematically more precise description of these concepts it is referred to Sections B.2 to B.4 in the appendix.
\subsection{Parameter Estimation}
So far we have defined the cumulative direct and indirect effects in equations (\ref{dir}) and (\ref{indir}), now we briefly want to recapitulate the estimation equations according to \cite{roysland}. Let us assume, we have $I$ independent observations of the above defined quantities, with individual values denoted as $N_i(t)$, $Y_i(t)$, $X_{1i}, X_{2i}(t), Z_{1i},\ldots, Z_{pi}$, $M_{i}(t)$ for $i=1, \ldots, I$ at time $t$ . Let further $\mathbf{L}(t)$ be the matrix with the $i$ th row being equal to $Y_i(t)(1,X_{1i}, X_{2i}(t-), Z_{1i},\ldots, Z_{pi}) $  and  $\mathbf{N}(t)$ and $\mathrm{d}\mathbf{M}(t)$ the vector consisting of the respective individual numbers of events $N_i(t)$ and martingale increments $\mathrm{d}M_{i}(t)$, as well as 
$\mathrm{d}\mathbf{B}(t)=(\beta_{3,0}(t)\mathrm{d}t,\ldots,\beta_{3,p+2}(t)\mathrm{d}t)^{T}$
 the vector of regression functions. \\
Then we can write (\ref{add_model}) in matrix notation
\begin{align*}
\mathrm{d}\mathbf{N}(t)=\mathbf{L}(t)\mathrm{d}\mathbf{B}(t) + \mathrm{d}\mathbf{M}(t),
\end{align*}
which suggests that, given that $\mathbf{L}(t)$ has full rank, an estimator for the local effects $\mathrm{d}\mathbf{B}(t)$, denoted as $\mathrm{d}\widehat{\mathbf{B}}(t)$, can be obtained by solving the estimation equations
\begin{align*}
\mathbf{L}(t)^{T}\mathrm{d}\mathbf{N}(t) - \mathbf{L}(t)^{T}\mathbf{L}(t) \mathrm{d}\mathbf{\widehat{B}}(t) = 0 
\end{align*}
for every event time. The estimates $\mathrm{d}\mathbf{\widehat{B}}(t)$ of the local effects may be quite variable. But we obtain stable estimates of the cumulative regression functions $\mathbf{B}(t)$ by aggregating the local effects over all event times up to and including time $t$. \\ 
For the mediator processes $\boldsymbol{X}_2(t)$ we assumed a linear model 
\begin{align*}
\boldsymbol{X}_2(t)=  \tilde{\mathbf{L}}(t)\boldsymbol{b}_2(t) + \boldsymbol{W_2}(t),
\end{align*}
where the $i$-th line of $\tilde{\mathbf{L}}(t)$ equals $Y_i(t)(1,X_{1i}, Z_{1i},\ldots, Z_{pi}) $ and $\boldsymbol{W_2}(t)$ has zero-mean, is independent of $\mathbf{L}(t)$ and its components are uncorrelated.\\
So the estimates $\boldsymbol{\hat{b}_2}(t)$ can be deduced from the standard normal equations
\begin{align*}
\tilde{\mathbf{L}}(t)^{T}\boldsymbol{X_2}(t)-\tilde{\mathbf{L}}(t)^{T}\tilde{\mathbf{L}}(t)\boldsymbol{\hat{b}_2}(t) = 0
 \end{align*}
 at each event time.
Now we can insert the obtained local estimates  $\mathrm{d}\hat{B}_{3,1}(t)$ and $\mathrm{d}\hat{B}_{3,2}(t)$ for $\beta_{3,1} (t)\mathrm{d}t$ and $\beta_{3,2}(t)\mathrm{d}t$ and the respective OLS estimate for $b_{2,1}(t)$ in equations  (\ref{dir}) and (\ref{indir}). Then the estimates of the direct and indirect cumulative effects become sums over the observed event times. For more details see \cite{fosen}.   

\section{Survival collider issue} \label{survival_collider}
In the previos section we have discussed how dynamic path analysis can be viewed as a series of local DAGs. However, when attempting to connect those DAGs as in Figure \ref{coll_2}, one question became apparent and that was whether survival in itself produces associations between covariates. In our case that would imply that the association between the treatment and the mediator could partly be due to survival selection, i.e. an artefact. Just to clarify, consider Figure \ref{coll_2}, where $X_1$ denotes the treatment, $X_2(s_0) $ the mediator value measured at baseline, say $s_0$ and further $ \mathrm{d}N(t_j)$ the jump of the counting process at time $t_j$. One could speculate that when estimating the respective effects at event time $t_2$ one would potentially introduce spurious association between $X_1$ and  $X_2(s_0) $ by conditioning on having survived event time $t_1$, which appears as a collider in the DAG.
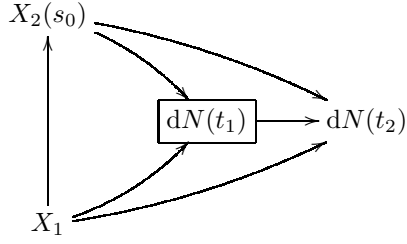
\begin{figure}[H]
\begin{displaymath}
\xymatrix{
 X_2(s_0) \ar@/^/[dr] \ar@/^/[drr] \\
 &*+[F-]{\mathrm{d}N(t_1)} \ar[r] & \mathrm{d}N(t_2)\\
X_1 \ar[uu] \ar@/_/[ur] \ar@/_/[urr]
}
\end{displaymath}
\caption{Simplest scenario where a survival collider effect could be suspected.}
\label{coll_2}
\end{figure}
The following argument indicates, however, that this is not a problem, so in that case conditioning on a collider (survival) does not produce artificial association. \\
We will first look at the case of two covariates that are independent at baseline and show that independence is preserved under survival assuming the additive hazard model. In a second step, we will consider covariates that satisfy a linear structural equation model and directly relate that case to the simple dynamic path analysis model illustrated in Section \ref{dynamic_path}.
\subsection{Independence of covariates is preserved under survival in the
additive hazard model} \label{independent_cov}
Let  $X_1$ and $X_2$ here denote two independent covariates measured at baseline. Under a simple additive hazard model with no interactions, the probability of surviving up to time $t$ is 
\begin{align*}
C(t)\exp{\left(-(\int_0^t{\beta_1(s)}\mathrm{d}s) X_1- (\int_0^t{\beta_2(s)} \mathrm{d}s) X_2\right)},
\end{align*}
where $C(t)$ results from taking the exponential of the integral over the baseline regression function.
 Let $T$ denote the survival time. Then
\begin{align*}
P(X_1&  =x_1, X_2=x_2|T>t)\\
&  =\frac{P(X_1=x_1,X_2=x_2,T>t)}{P(T>t)}\\
&  =\frac{P(T>t|X_1=x_1,X_2=x_2)f_{X_1}(x_1)f_{X_2}(x_2)}{P(T>t)}%
\end{align*}
where $f_{X_1}(x_1)$ and $f_{X_2}(x_2)$ are the densities of $X_1$ and $X_2$ at time 0. Inserting the survival probability gives us
\begin{align*}
P(X_1  &  =x_1,X_2=x_2|T>t)\\
&  =\frac{C(t)\exp{\left(-(\int_0^t{\beta_1(s)}\mathrm{d}s) x_1-(\int_0^t{\beta_2(s)} \mathrm{d}s) x_2\right)}f_{X_1}(x_1)f_{X_2}(x_2)}{P(T>t)}.%
\end{align*}
Hence the probability can still be factorized, so $X_1$ and $X_2$ are still independent at time $t$ conditional on survival.\\
The result can immediately be generalised to an additive model with $K$ independent covariates and no interactions.
\subsection{Covariates satisfying a linear structural equation model}
The above statement can be generalized to a very useful result about how the
distribution of variables in a linear structural equation model are changing
under survival selection. Let us define an ordered linear structural equation
model as follows (modification of \cite{pearl_book}, equation (1.41), see also \cite{loh})
\begin{align*}
X_{k}= b_{k,0}+\sum_{j<k}b_{k,j}X_{j}+W_{k}.
\end{align*}
This is defined for $k=1,\ldots,n$, and all $W_{k}$ are assumed to be independent. Note that we assume this model to be structural, hence there are no unmeasured confounders.
In matrix form the model can be written as follows with the solution in
$W_{k}$ on the right:
\begin{align*}
\mathbf{X}=\boldsymbol{b_0}+\mathbf{B}\mathbf{X}+\boldsymbol{W},\qquad \mathbf{X}=(\mathbf{I}-\mathbf{B})^{-1}\boldsymbol{ b_0}+(\mathbf{I}-\mathbf{B})^{-1}\boldsymbol{W}.
\end{align*}
Here, $\mathbf{B}$ is a strictly lower triangular matrix.
Assume that the hazard function is a linear combination of the $X$'s, then it
is also a linear combination of the $W$'s, and hence the independence
of the $W$'s\ are preserved conditional on survival up to a given
time $t$ because of Section \ref{independent_cov}.
Let $E_{t}$ denote conditional expectation given survival up to time $t$, that
is, $T>t$, then we have 
\begin{align}
E_{t}[X_{k}-b_{k,0}-\sum_{j<k}b_{k,j}X_{j}|X_{j},j<k]=E_{t}[X_k|X_j,j<k]-b_{k,0}-\sum_{j<k}{b_{k,l}X_j}\label{1line}
\end{align}
Since the independence of the $W$'s\ are preserved under survival,
the left hand side of the above equation can be written in the following way:
\begin{align*}
E_{t}[X_{k}-b_{k,0}-\sum_{j<k}b_{k,j}X_{j}\,|\;X_{j},j  & <k]=E_{t}[W
_{k}|W_{j},j<k]=E_{t}[W_{k}]\\
& =E_{t}[X_{k}-b_{k,0}-\sum_{j<k}b_{k,j}X_{j}]=E_{t}[X_{k}]-b_{k,0}-\sum_{j<k}{
b_{k,j}E_{t}[X_{j}]}.
\end{align*}
Comparing with (\ref{1line}), we have:
\begin{align*}
E_{t}[X_{k}|X_{j},j<k]=\sum_{j<k}{b_{k,j}X_{j}}+E_{t}
[X_{k}]-\sum_{j<k}{b_{k,j}E_{t}[X_{j}]}.
\end{align*}
Hence, given survival at time $t$ we have a linear model with the same coefficients that we had at time 0. Hence there is a stable linear relationship between the variables under survival selection. The constant term changes, but that does not matter.
\subsection{Relation to a simple model for dynamic path analysis}
Let us explicitly point out how the results above relate to dynamic path analysis. Assume that $X_1$ is treatment and $X_2(s_0)$ is the mediator, both measured at time
zero. Assume that these covariates constitute a simple structural equation
model, that is 
\begin{align*}
X_1=b_{1,0}+W_{1},\quad X_2(s_0)=b_{2,0} (s_0)+ b_{2,1}(s_0) X_1+W_{2}(s_0).%
\end{align*}
There are no unmeasured confounders either for the relationship between $X_1$ and
$X_2(s_0)$ nor for their influence on survival which is determined by the additive
hazard model:
\begin{align*}
\alpha(t)= \beta_{3,0}(t)+\beta_{3,1}(t)X_1+\beta_{3,2}(t)X_2(s_0).
\end{align*}
This is assumed to be a structural model. Given survival at any given time,
say $T>t$, the relationship between $X_1$ and $ X_2(s_0)$ is still determined by the
coefficient $b_{2,1}(s_0) $. Hence, a simple calculation of local direct and indirect effects of
treatment $X_1$ on hazard would be:%
\begin{align*}
\text{\textrm{Direct effect: \ \ \ \ }}\beta_{3,1}(t)
\end{align*}
\begin{align*}
\text{\textrm{Indirect effect: \ \ \ \ }} b_{2,1}(s_0)\beta_{3,2}(t)
\end{align*}
Note, that the same argument would locally apply, once we obtain an updated measurement of the mediator, say $X_2(s_1) $, and we would consider estimation of the effects around the path plotted as dashed lines in Figure \ref{col_simple}. It could occur that the estimate for the treatment-mediator relationship  $b_{2,1}(s_1)$, differs from $b_{2,1}(s_0)$. This difference, however, will be due to an actual change in the effect that can be observed between the treatment and the updated measurement of the mediator but not due to any selection bias. Referring once more to our real data application, that would for example mean, that the effect of treatment on LDL-cholesterol levels indeed attenuates after some time. \\
\begin{figure}[H]
\begin{displaymath}
\xymatrix{
 X_2(s_0) \ar@/^/[dr]    
  &   &  X_2(s_1)  \ar@/^/[dr] \ar@{-->}@/^/[drr]\\\
 &*+[F-]{\mathrm{d}N(t_1)} \ar[rr] &  &*+[F-]{\mathrm{d}N(t_2)}  \ar[r]  &\mathrm{d}N(t_3)\\
X_1 \ar[uu] \ar@/_/[ur]\ar@{-->}@/_8mm/[uurr] \ar@/_/[urrr] \ar@{-->}@/_/[urrrr]
}
\end{displaymath}
\caption{Local path of interest for the estimation of direct and indirect effects, when an updated measurement became available.}
\label{col_simple}
\end{figure}
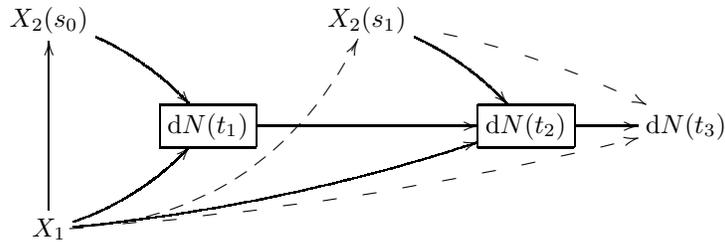
\subsection{Collapsibility in the additive model}
An important question is whether a factor that influences survival, but is un-related to other covariates, can become a confounder over time. Let us assume there is an unmeasured quantity, that is not considered a confounder at time 0 since it is only associated with survival and not with measured covariates and there are no interactions present. It can be a part of the structural equation model above, but with all the $b$'s equal to zero. It follows from the previous result that the $b$'s will stay at zero conditioned on survival. Hence, unmeasured confounders do not arise due to survival collider effects. If the Cox model is valid instead, this will not be true. In fact, any conceivable factor influencing survival will become a confounder when time passes. The underlying reason for or the different behavior of the two models arises from the collapsibility of the additive model (as discussed in \cite{collaps}), which does not hold true for the Cox model.

\section{Simulation study}\label{simulation_study}
To further motivate investigations considering the mediator to be a time-varying process rather than a time-fixed covariate, we performed a simulation study.
We compared dynamic path analysis using a single mediator measurement to dynamic path analysis using several measurements and also set those results in relation to the true data generating model.
In the subsequent section we describe our data generating strategy for the survival outcome and mediator process and then show results from various simulation scenarios. All simulations and analyses were performed in R, version 3.0.1 \cite{R_prog}
\subsection{Data generation}
To generate data respecting the dynamic path analysis setting as presented in Section \ref{dynamic_path} with structural equations as 
\begin{align*}
X_1 &= b_{1,0}+W_1\\
X_2(s) &= b_{2,0}(s) + b_{2,1}(s)X_1 + W_2(s)\\
\alpha(s) &= \beta_{3,0}(s) + \beta_{3,1}(s) X_1 + \beta_{3,2}(s) X_2(s), 
\end{align*}
 we proceeded the following way:\\
Let us first consider the simulation of the time-to-event data. Generally, we assumed a study period of 5 years. This time period was discretised into equidistant grid points, where the distance, $\Delta$, was set to $1$ week. 
For each interval of size $\Delta$ the probability of an event within that interval can be expressed as
\begin{eqnarray*}
P(t_h< T \leq t_h + \Delta|T>t_h)=1-\exp{\left(-\int_{t_h}^{t_h + \Delta}{\alpha(s)\mathrm{d}s}\right)}.
\end{eqnarray*}
Making use of the Taylor expansion this probability can be approximated by 
$\alpha(t_h) \Delta$
and we can, in turn, apply the additive form of the hazard $\alpha(t_h)$ including a time fixed or a time-dependent mediating variable. Note beforehand, that the choice of values in the subsequent paragraph might seem rather arbitrary at first glance, but the underlying idea was to create a situation that appears quite similar to the real data example presented in the next section.\\
To model the flexible regression functions $\beta_{3,j}(s)$, we applied cubic splines and specified the function to take the values $(0.04, 0.03, 0.02, 0.02)$ at years $(0,1,3,5)$ for the effect of the mediating variable on the hazard and to take the values $(-0.3, -0.1, -0.6, -0.05, -0.05, -0.05)$ at years  $(0, 0.2, 0.8, 1.1, 3.5, 5)$ for the effect of treatment. 
To model a $50:50$ randomisation, the treatment variable was simply sampled from a Bernoulli distribution. For the mediator process we first generated a baseline value, sampled from a normal distribution with mean 11 and standard deviation 1.5. Again, we assumed a time-varying effect of treatment on the mediator and modeled the time-dependent coefficient with a cubic spline function (specified at times $(0,1,2,3,4,5)$ to take values $(-0.1,-3,-2.2,-3.3,-2.9,-2.9)$)  and added the time-varying treatment component to the sampled baseline values of the mediator and additionally a noise term generated from a multivariate normal distribution with mean vector equal to zero and a covariance matrix with diagonal elements set to $0.05$ and non-diagonal number set to $0.0$ .\\
After generating a treatment indicator, the mediator process and respective time-varying regression functions $\beta_{3,j}(s)$ at each discrete time point $t_h$, we could insert the respective values into the function for the hazard and calculate event probabilities for each individual at each discrete time point. This enabled us to generate an event history for each individual by sampling from a Bernoulli distribution with the calculated event probability at each time point. If an individual was sampled to have an event at one particular time $t_h$ its event indicator was set to one and the survival time to the discrete time $t_h$ where the event occurred. Censoring times were generated from a uniform distribution to account in average for a $13 \%$ censoring rate. Individuals still alive and not censored at year 5 were censored at that time point due to end of study.
\subsection{Simulation scenarios}
As stated above the main objective was to investigate what one could potentially gain from using more than just one measurement of the mediator, assuming that the mediator is actually a time-varying process. So a natural scenario to start with was, to just take one 'snap-shot' measurement from our generated mediator process and only utilize that measurement in the dynamic path analysis model. Our flexible simulation procedure, however, allowed us to take 'snap-shot' measurements at various time points and to further make use of them in our model for analysis. For the 'snapshot' measurements we simply picked values of the simulated mediator process at pre-specified time points for every individual still in the study at those time points. Again the time points were chosen to mimic the real data application. 
\subsection{Simulation results}
In Figures \ref{sim1} to \ref{sim3} results based on 100 simulations based on trials of sample size 2000 are presented. The scenarios vary in how many 'snap-shot' measurements of the mediator were taken. In all plots the true curve, displaying the assumed parameters for our simulations, is plotted as a dotted blue line, the curve obtained from dynamic path analysis only using one mediator measurement is represented as a red dotted line, whereas the curve from dynamic path analysis using several mediator measurements is displayed as a black solid line.
\begin{figure}[H]
	\centering
  \includegraphics[scale=0.45]{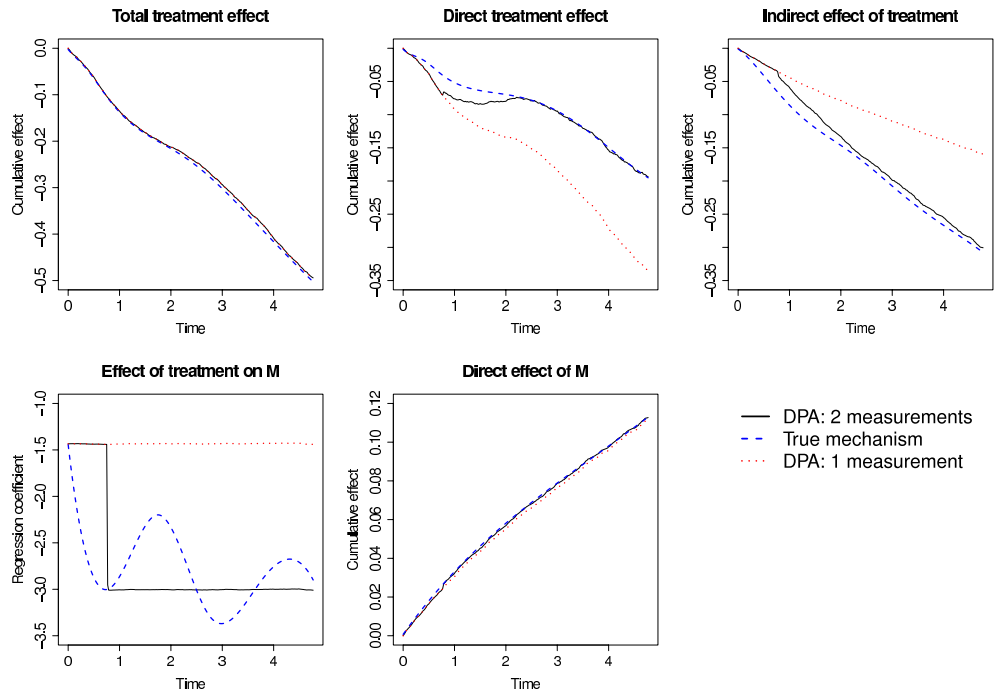}
	\caption{Simulation results, comparing a dynamic path model (DPA) only utilizing a single measurement compared to a model using two measurement and the true model.}
	\label{sim1}
\end{figure}
\begin{figure}[H]
	\centering
  \includegraphics[scale=0.45]{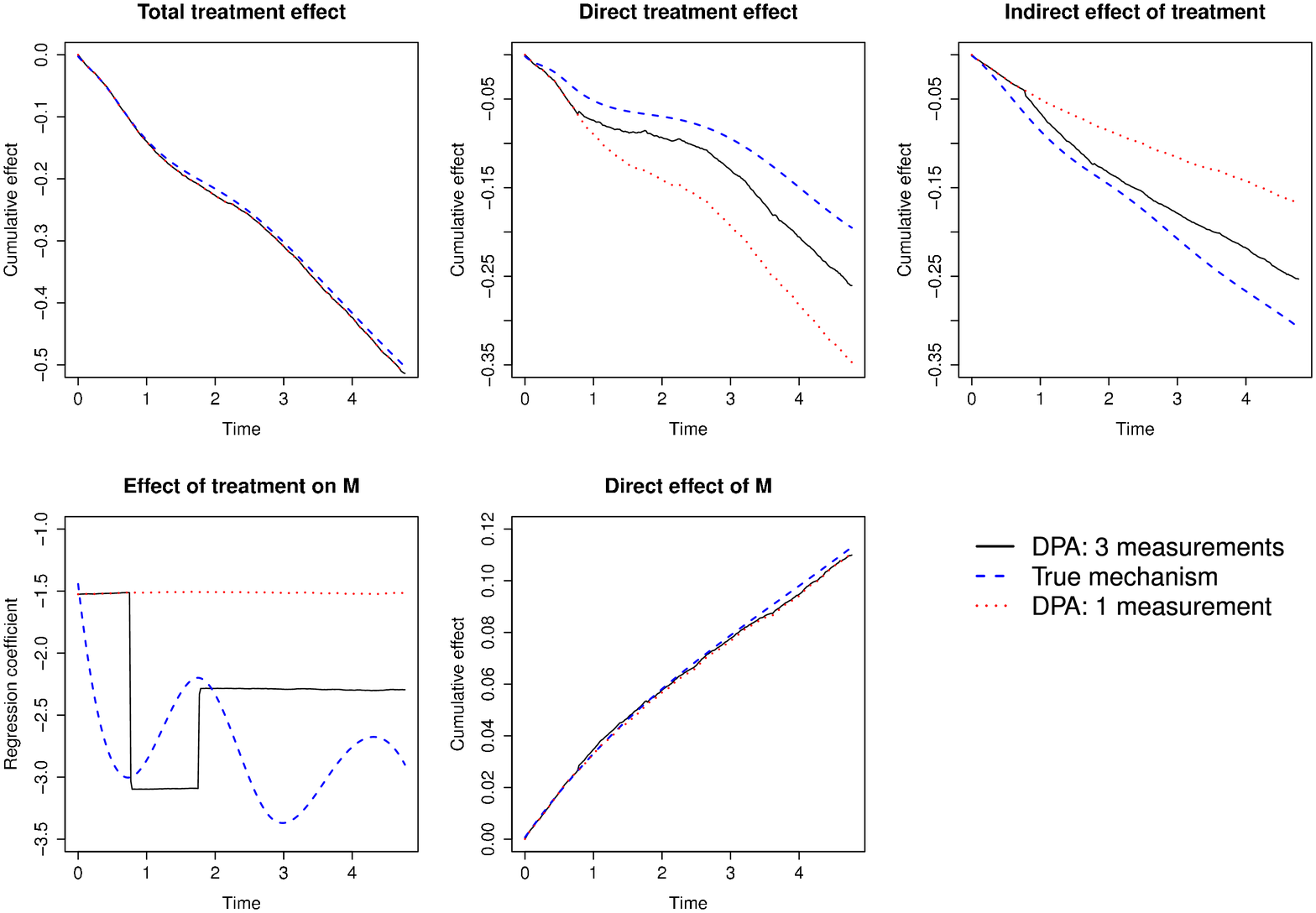}
	\caption{Simulation results, comparing a dynamic path model (DPA) only utilizing a single measurement compared to a model using three measurement and the true model.}
	\label{sim2}
\end{figure}
\begin{figure}[H]
	\centering
  \includegraphics[scale=0.45]{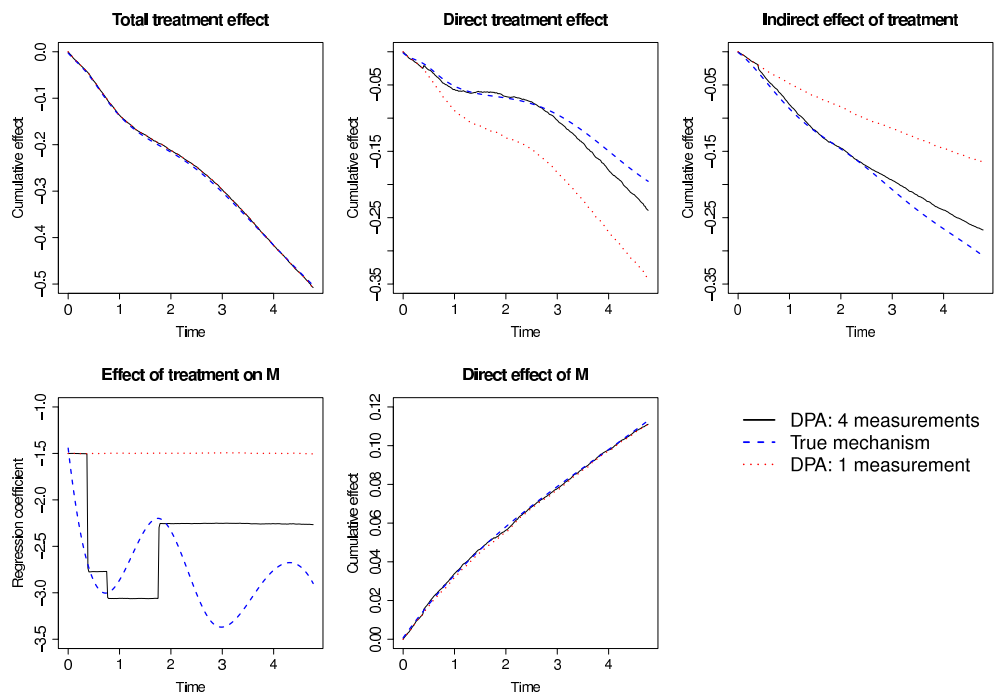}
	\caption{Simulation results, comparing a dynamic path model (DPA) only utilizing a single measurement compared to a model using four measurement and the true model.}
	\label{sim3}
\end{figure}
What can be observed throughout all considered scenarios is a tendency to underestimate the indirect effect while overestimating the direct effect of treatment on the survival outcome when only a single measurement is used. However, it also appears to be important at which particular time point or state of the mediator process the measurements are taken. For example, comparing Figures \ref{sim1} and \ref{sim2} one can see that the model using 2 measurements gives actually better results for most of the study period because the measurement of the mediator was taken at a time point where the treatment effect on the mediator roughly equaled the average treatment effect of the remaining study period. A similar pattern can be observed comparing Figures \ref{sim1} and \ref{sim3}. For the study period after year 2, both direct and indirect effects tend to deviate more from the truth in Figure \ref{sim3}. However, it can also be observed that having more frequent measurements before year 2 makes it by far more likely to actually capture the nature of the underlying process. 
\section{Analysis of the IDEAL data} \label{IDEAL_analysis}
In the following section we want to report the results from application of various dynamic path models to the data collected within the IDEAL trial. 
Again, we primarily want to focus on comparing models using a different amount of available mediator measurements. 
\subsection{The IDEAL study in brief}
Recall, that we are dealing with a secondary prevention trial, that aimed to compare two different statin treatment strategies (high dose of atorvastatin and standard dose of simvastatin) to lower cholesterol levels in patients with previous myocardial infarction. Patients eligible to the study (for detailed description of respective criteria see \cite{ideal_study_design}) were randomized to receive either 20 mg of simvastatin or 80 mg of atorvastatin, with a foreseen possibility to change dose at week 24. Lipoprotein levels (low density lipoprotein (LDL) - cholesterol, high density lipoprotein (HDL) - cholesterol, total cholesterol (TC), triclycerides (TG), apolipoprotein B-100 (Apo B) , apolipoprotein A (Apo A))  from fasting blood samples along with lever enzymes and other laboratory measurements were taken at baseline, at 12 and 24 weeks, 1 year and yearly thereafter.\\
The primary endpoint was defined as time to first occurrence of a major coronary event (CHD). However, to exploit more of the collected outcome information, we considered one of the broader secondary endpoints, time to 'any' coronary heart disease event, since - as mentioned above - our main focus concerns the amount of mediator measurements used. (For more detailed definition of the endpoints see \cite{ideal_main}.)
\subsection{Patient selection for the present analysis}
From the originally 8888 randomised patients, after careful considerations, 8646 patients were left for our analyses, as we required a certain amount of information to see any mediation effects in our application of dynamic path analysis. More specifically, this means that we excluded patients that had no baseline information on the considered mediator (described below) and also those still under risk at week 12 but with missing mediator values, as those values are needed to actually observe an effect of treatment. If successive measures of the time-dependent covariate were not on hand, the last observation was carried forward.
\subsection{Event frequency and deviations from the assigned treatment regime}
From the previously mentioned 8646 patients, 4350 patients were originally randomised to 20 mg of simvastatin and 4296 to 80 mg atorvastatin, where of 1019 patients experienced the event of interest in the simvastin group and 852 in the atorvastatin group.  
However, the study design included the option to adjust a patient's dosage at week 24. As a result of these possible adjustments a fairly large number of patients deviated from their allocated treatment at least once throughout the study period. In the present paper, we decided to focus on the effects of allocating the described intervention in practice and therefore present our results from dynamic path analysis applying the intention-to-treat (ITT) principle, analysing every individual as randomised.\\
Estimating effects comparing specific treatments at a specific dose would involve a more complicated analysis due to the special study design and is therefore subject to further work on its own. 
\subsection{Potential Mediators}
The a priori understanding of the underlying mechanisms made LDL-cholesterol the logical candidate for a mediator. However, due to the ongoing debate in clinical literature, whether LDL-cholesterol or apolipoprotein B (Apo B) should be considered as the target value for lipid therapy (see e.g. \cite{LDL_APO_B}), we also took Apo B into consideration as a potential mediator. Since Apo B and LDL-cholesterol are highly correlated and it is still unclear in which way statin treatment effects Apo B, we refrained from applying a more complex joint model and run separate analyses for LDL-cholesterol and Apo B. To not to distract from the main points we indented to highlight, we only present the development of the proportion of the effect mediated over time for LDL-cholesterol and want to refer to supplementary plots of the actual path effects in the Apo B models. \\
\subsection{Results}
Throughout the IDEAL study blood lipids were measured at 7 measurement points after baseline. So we decided to compare models where all available repeated measurements were utilized compared to a reduced analysis were we only utilized the baseline and  week 12 measurements of the mediator and discarded all measurements after week 12.
The reason for using the week 12 measurement was just because one would not see any effect of treatment on the mediator and therefore any indirect effect directly at administration of treatment. \\
In all models presented here the square root of LDL-cholesterol [mg/dL] was considered as the time-dependent mediating variable $X_2 (t)$.  The models for the treatment mediator relationship as well as the treatment outcome relationship included LDL-cholesterol at baseline (square root transformed), APO B level [g/L] at baseline (square root transformed), high density lipoprotein [mg/dL] at baseline (square root transformed) as well as smoking status, usage of statin treatments at randomisation, and furhter presence of hypertension, presence of diabetes, sex and age at baseline, as those covariates were reported to affect the risk of CHD events and also the LDL-cholesterol levels \cite{chd_risk_fact,ldl_risk_fact} .\\
In the left-hand column of Figure \ref{total} the estimated cumulative total, direct and indirect effects of atorvastatin compared to simvastatin mediated through LDL-cholesterol are shown using all available repeated measurements applying the ITT principle. In the right-hand column of Figure \ref{total} one can find the same arrangement of plots but for the analysis where only the baseline and week 12 measurements of the mediator was used. The week 12 measurement was carried forward for that analysis, mimicking a situation where one would only have access to one measurement. \\
In Figure \ref{direct_lm} the corresponding linear regression coefficients and the cumulative direct effect of LDL on the event outcome are displayed for these two models. 
Generally the patterns for the effect on LDL as well as the results for the indirect effect fit the intuitive understanding quite well, that statins at a higher dose have a rather rapid effect on LDL-cholesterol levels after receiving treatment, but it takes some time until these effects actually effect the risk of cardiovascular disease.\\
Comparing the plots of the treatment effect on LDL-cholesterol in Figure \ref{direct_lm} one can clearly see that the effect of atorvastatin compared to simvastatin on the mediating variable actually attenuates over time. By only utilizing the baseline and week 12 measurements, however, a constantly greater treatment effect on the mediator is assumed. Comparing the plots for the indirect effects one can observe that the indirect effect appears less pronounced when only the baseline and week 12 measurements are incorporated in the analysis. This could be due to the general tendency to underestimate the indirect effect that was also observable in the simulation study. Likewise, the direct effect of treatment behaves similar to our simulation results, where could observe a tendency for overestimation. Attempting to contribute to the ongoing discussion in clinical literature, we show the development of the proportion of effect mediated through LDL-cholesterol and Apo B in different plots in Figure \ref{prop_ldl}. It can be observed that the proportion of the effect mediated through Apo B has a more pronounced peak within the first year and also stabilizes at a higher level thereafter. 
\begin{figure}
\begin{center}
\includegraphics[width=0.8\textwidth]{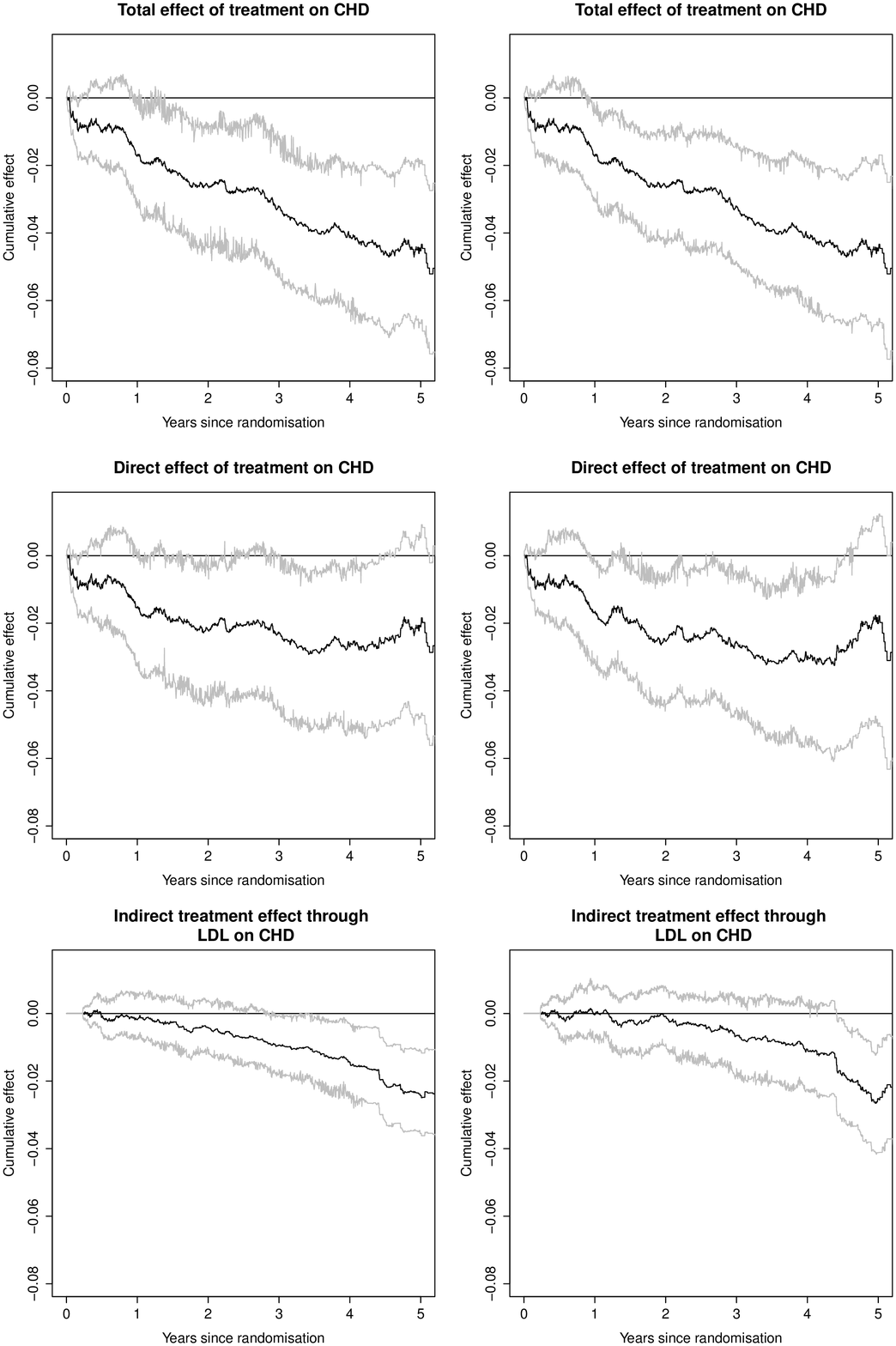}
\end{center}
\caption{First column: Estimated cumulative total, direct and indirect effect of atorvastatin compared to simvastatin mediated through LDL cholesterol using all available measurements; Second column: Effect estimates only utilizing the week 12 measurement; Grey lines represent the 95\% confidence intervals based on 200 bootstrap samples. }
	\label{total}
\end{figure}

\begin{figure}
\begin{center}
\includegraphics[width=0.75\textwidth]{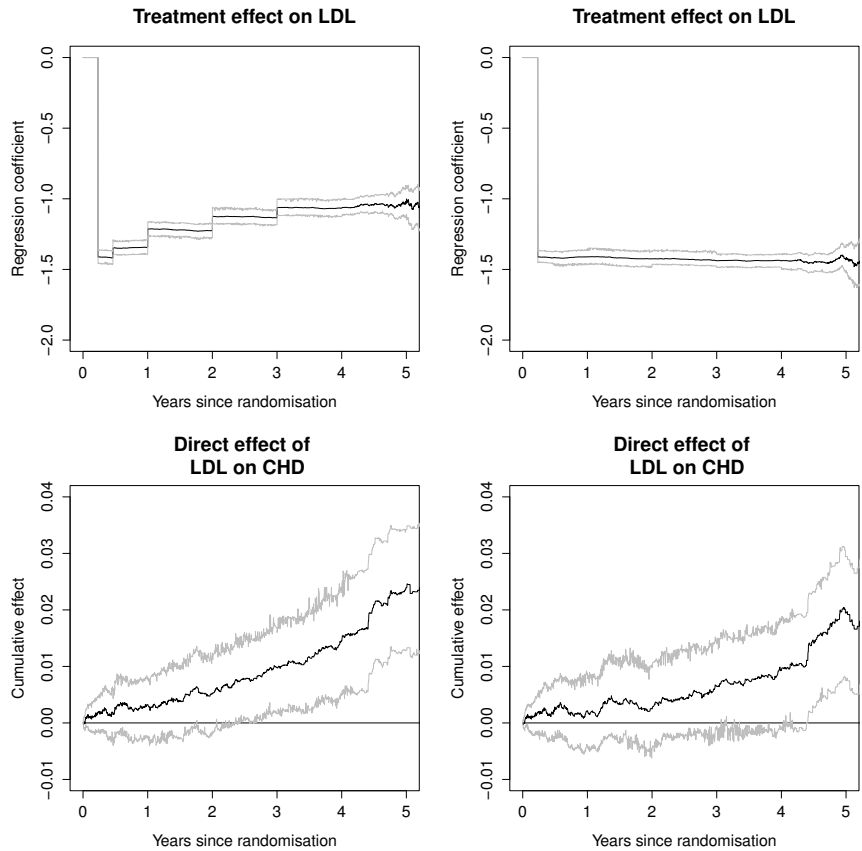}
\end{center}
\caption{First column: Regression coefficients over time for the effect of atorvastatin compared to simvastatin on LDL cholesterol and cumulative direct effect of LDL cholesterol on CHD using all available measurements; Second column: Effect estimates only utilizing the week 12 measurement; Grey lines represent the 95\% confidence intervals based on 200 bootstrap samples.}
\label{direct_lm}
\end{figure}
\begin{figure}
\begin{center}
\includegraphics[width=0.75\textwidth]{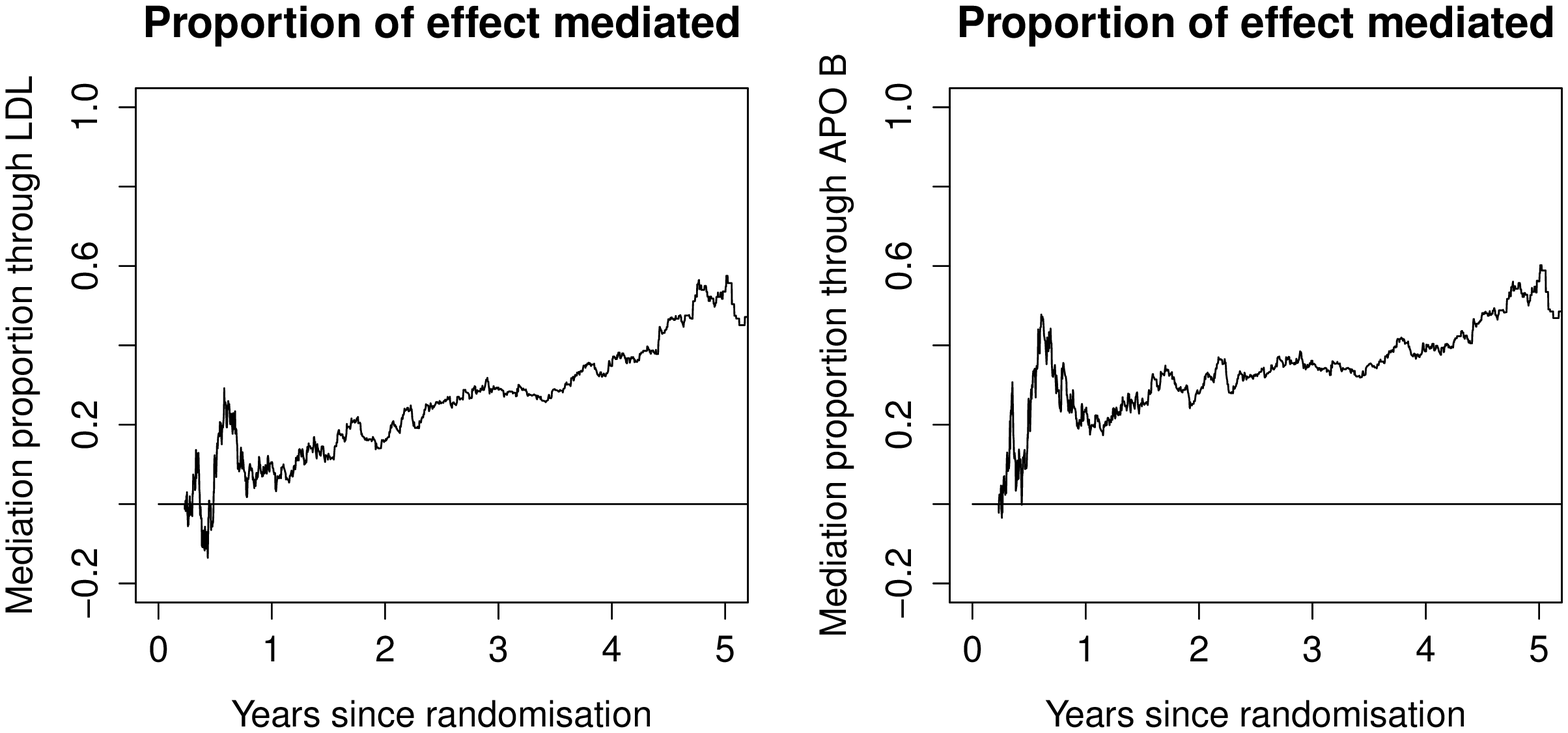}
\end{center}
\caption{Proportion of treatment effect mediated through LDL-cholesterol vs. Apo B. }
\label{prop_ldl}
\end{figure}

\section{Discussion}\label{discussion}
In this paper, we mainly addressed two issues around the model of dynamic path analysis.
Most importantly we could show, that conditioning on the collider 'survival' does not produce artificial associations between treatment and the mediator. This implies that whenever one finds the effect between treatment and the mediator to be changing over time, that phenomenon corresponds to an actual change in the relationship and is not spuriously introduced by survival selection. From our given argument the important notion follows, that unmeasured quantities that are not considered to be confounders but affect survival in an additive manner, will not turn into confounders as time passes by. A property that will not hold true for the Cox model. This is a consequence of the fact that the additive hazard model is collapsible, which is not the case for the Cox model \cite{collaps}. \\
Other works around the dynamic path analysis model investigated specific confounding scenarios, but only considered a time-fixed mediator \cite{martinussen2,lange}. In our data example we, however, assumed no unmeasured confounding between exposure and outcome as well as between the mediator and the outcome and aimed at a causal interpretation. But we are aware that the estimated cumulative direct, indirect and total effects have to be interpreted with caution, as these assumptions are generally not testable. \\
Concerning the results from our simulation study and our analyses of the IDEAL data we mainly want to emphasis one issue that has, for example, been mentioned in \cite{believe}. That is, when the underlying mediating mechanism is actually working in time, performing mediation analyses that only incorporate one single mediator measurement can distort the real picture and one may risk to deduce false conclusions from such analyses. \\
Looking at our presented plots, a pattern that could be observed all over was, that the bootstrap confidence intervals for the cumulative indirect effect appeared narrower than those for the direct effects and consequently for the total effect. This phenomenon could also be observed in the applications presented in \cite{gamborg}. We speculate that that could be due to measurement error in the assessment of LDL-cholesterol and Apo B, which was reported to be a problem for multi-center studies \cite{contois_lipids}. \\
Note that there are several other interesting challenges around the IDEAL data one could pick for more in-depth analyses. For example, the issue of non-compliance with assigned treatment, which had already been addressed in \cite{ideal_adherence}. One possible extension would be to employ inverse probability of censoring weights to and further perform reweighted dynamic path analysis model, as suggested in \cite{roysland}.\\
From a practical perspective estimates of effects over time could improve study planning of future trails with similar objectives. For example by considering the mediator as a surrogate endpoint, the respective plots of the cumulative indirect effect could help to preplan the timing of interim measurements.\\
In summary, dynamic path analysis is a useful tool for mediation analysis with a time-dependent mediator and a survival outcome. The obtained cumulative effect plots enable us to describe how direct and indirect effects evolve over time, which can add to the mechanistic understanding of the underlying processes. The linearity of the models for the treatment-mediator relationship as well as for the hazard appears to have appealing properties.
\newpage
\begin{appendix}
\section*{ APPENDIX: Causal interpretation in dynamic path analysis}
\section{The statistical model}
        We will consider scenarios that fit the following combined longitudinal
        and time-to-event model: 
        \subsection{Covariates} The main event of
        interest occurs at time $T > 0$, represented by a counting process
        $N_t$.  In addition, there are variables $X_1, \dots, X_n$ that can
        have outcomes before this event occurs.  More precisely, we assume that
        the outcome of each variable $X_i$ occurs at a deterministic time $t_i$
        if the main event has not already occurred, i.e. if  $T > t_i$.  We
        also assume that $i < j$ means that the outcome of $X_j$ can not occur
        strictly before $X_i$.  This means that we have deterministic times
        $t_1 \leq \dots \leq t_n$, possibly with repetitions.                            
          
        Let $P$ denote the joint density that governs the
        pre-intervention frequencies of these variables, and
        let $b_1, \dots ,b_n$ be functions on the form 
        \begin{align}
                b_k (X_1, \dots, X_{k-1}) = b_{k,0} + b_{k,1} X_1 + \dots + 
                b_{k,k-1} X_{k-1}.  
          \end{align}
        Assume that 
		\begin{align} \label{eq:ci}
         X_k - b_k(X_1, \dots, X_{k-1})  \ci
        \{X_1, \dots, X_{k-1} \} ~| ~T \geq t_k,
         \end{align}
        and 
       \begin{align} \label{eq:ce} 
        E [ X_k | X_1, \dots, X_{k-1}, T \geq t_k ] =  b_k(X_1, \dots, X_{k-1}) ,
  \end{align}
for every $k>1$ and let 
\begin{align}
b_{1,0}=E [X_1].
\end{align}
        \subsection{Covariate processes and additive hazards}                       
        The collection of possible events that could have occurred before $t$,
        the so-called history or filtration generated by $N_t$  and $\{ X_i
        |~  t_i \leq t\}$ will be denoted $\mathcal F_t$.  We will also need
        to consider the history restricted to $N_t$ and $\{ X_i |~  t_i \leq t,
        i \leq j\}$, which we will refer to as  $\mathcal F_t^j$.
        For mathematical convenience in the subsequent derivations, let us also introduce the following caglad function, $s \mapsto J_s$  from $[0,\infty)$ into the $n\times
        n$-matrices such that the only non-zero entries of $J_s$ are
        contained in the upper-left $n_s \times n_s$-corner, where $n_s :=
        \max{\{k | t_k \leq s\}}$.  
This gives us a vector-valued 'covariate
        process' $s \mapsto J_s X$ that is adapted to the filtration $\mathcal
        F_s$.                                
                    
        Finally, we assume that there exist functions $\beta_s^0$ and $\beta_s
        := (\beta^1_s, \dots, \beta_s^n)^\intercal$ such that
         \begin{align}
                \label{graphHaz} \alpha_s = \beta^0_s + \beta_s^\intercal  J_sX 
            \end{align}
                 defines the hazard of $N_s$ with respect to
        the pre-intervention joint density $P$ and the filtration $\mathcal
        F_s$. \\
Let us illustrate these concepts with the following example. Suppose we have a situation with a baseline variable and a 'process' that takes the first value at baseline, while it takes a new value at a later time, say $t$. Further, suppose the hazard of $T$, at time $s$, only depends linearly on the baseline variable and the current value of the process. Now, let $X^1$ denote the variable, $X^2$ denote the baseline value of the process and $X^3$ the subsequent value of the process. 
This gives 
\begin{align*}
J_s:=
\begin{pmatrix}
1 & 0 &0 \\
0 & I(s \leq t) & I(s>t)
\end{pmatrix},
\end{align*}
so for the covariate process we have
\begin{align*}
\begin{pmatrix}
X_1(s) \\
X_2(s)
\end{pmatrix}
= 
J_s 
\begin{pmatrix}
X^1\\
X^2\\
X^3
\end{pmatrix},
\end{align*}
so that the hazard can be written as
\begin{align*}
\alpha_s = \beta^0_s + \beta_s^\intercal  J_s 
\begin{pmatrix}
X^1\\
X^2\\
X^3
\end{pmatrix}.
\end{align*}

\subsection{Constant regression coefficients}
Due to linearity, we can estimate the coefficients $\{b_{k,j} \}_{ j> 0}$ 
by  performing ordinary linear regressions among the survivors at any time
after $t_k$, i.e.  
\begin{prop} \label{prop:constant}
    If we restrict to outcomes such that $T > t \geq t_k$, then  we have 
    \begin{align}
            E\big[ X_{k} \big|  \mathcal F_t^{k-1}  \big] = 
            \sum_{j = 1}^{k-1} b_{k,j} X_j +  E \big[  W_k  
           \big| T >t \big],  
    \end{align}
    where $W_k :=  X_{k} - b_k   ( X_1, \dots, X_{k-1})   $.
\end{prop}

        \section{Causal inference}

        \subsection{Causal validity}       
      We will say that the model is causal if, for every $V = (V_1, \dots,
      V_n)$, we can specify interventions that would have given  joint
      densities $\tilde P$ such that \eqref{eq:ce} and \eqref{eq:ci} hold when
      each function $b_k$ is replaced by $b_k + V_k$, while $\alpha$
      still defines a hazard with respect to $\tilde P$. If the model is causal, we are
      able to calculate the causal effects from these interventions on the
      marginal hazard of $N$. 

        \begin{thm} \label{thm:main}
             Let  $\alpha_t^{0}$ and $\tilde \alpha_t^{0}$  denote the marginal hazards of $N_t$
for $P$ and $\tilde P$  respectively. We have that 
\begin{align}\label{eq:mainEffect} 
       \tilde \alpha_t^0 =  \alpha^0_t  +  \sum_i \big(\beta_t ^\intercal J_t
       \big)_i V_i 
   + \sum_{ l> i } \big(\beta_t^\intercal  J_t \big)_l \sum_{p}  b_{i_k, i_{k-1}}
       b_{i_{k-1}, i_{k-2}} \dots   b_{i_2, i_{1}}V_i ,  
\end{align}
where we sum over partitions 
$p = (i_1 = i < i_2 < \dots < i_{k-1} < i_k = l)$, which correspond to the paths defined in Fosen, Section 2.5 \cite{fosen}.
        \end{thm}
        
        \subsection{Mediation}
        Suppose we could perform an intervention that would only perturb 
        the component $X_i$, i.e. we obtain a joint density $\tilde P$, where
        only the function $b_i$ is replaced by $b_i + \varepsilon$.

                              We see from \eqref{eq:mainEffect} that the
                                      effect seems to 
                                      propagate through the system. 
                To isolate the effect that is mediated through a variable 
                        $X_j$, where $j > i$, we could instead 
                        manipulate the system such that  $X_j$ behaved as if
                        $X_i$ had been perturbed, while we let $X_i$
                        remain unperturbed. 
                 
                        Note that we have
                        \begin{align*}
                                 b_j( X_1, \dots, X_{i-1},X_i + \epsilon, 
                                X_{i+1}, \dots, X_{j-1} )  =  
                                b_j(X_1, \dots, X_{j-1}) + b_{j,i} \cdot
                                \varepsilon. 
                        \end{align*}
               This means that if the model is causal,  
               then we could intervene such that 
                        $b_j(X_1, \dots, X_{j-1})$ is replaced by 
                        $b_j(X_1, \dots, X_{j-1}) +  b_{j,i} \cdot \varepsilon$
                        to obtain a scenario as we just described. 
                 Therefore, by Theorem  \ref{thm:main}, the mediated 
                 effect on the marginal hazard of $N$                     
                 is given by 
                         \begin{align} \label{eq:medEffect}
        \sum_{ l> i } \big(\beta_t J_t \big)_l \sum_{p}  b_{i_k, i_{k-1}}
       b_{i_{k-1}, i_{k-2}} \dots   b_{i_2, i_{1}}  b_{j,i} \varepsilon,  
\end{align}
where we sum over partitions 
$p$ such that $p= (i_1 = j < i_2 < \dots < i_{k-1} < i_k = l)$. 
This resembles the classical definition of natural indirect effects that builds on nested
counterfactuals, see \cite{pearl_book}, \cite{robins_greenland_identifiability} and \cite{lange}
Note that we are able to define indirect effects without explicit use of 
nested counterfactuals by exploiting the linearity. It is an interesting
observation that this, at least approximately, gives meaning for non-linear models
where $b_1, \dots, b_n$ are replaced by smooth-functions, since smooth
functions behave as if they are affine for small $\varepsilon$. 
\subsection{Causal pathways and path-effects}
  Suppose we are given a pathway: 
                \begin{align} 
                        \label{eq:path}
                        X_{i_1} \rightarrow X_{i_2} \rightarrow 
        \dots \rightarrow X_{i_k} \rightarrow N_t,
                \end{align}
                where $i_j < i_{j+1}$. 
        If the model is causal with respect to all the nodes $X_{i_1},
                \dots, X_{i_k}$, then we can isolate the effect 
                that flows only through this path by 
                first isolating the effect that is mediated by $X_{i_2}$, and
                then iterate this contrast for each $i_j$ such that $1 < j \leq
                k$. This gives the following path-effect:
                \begin{align} \label{eq:pathEffect}
   \big(\beta_t J_t \big)_{i_k} b_{i_k, i_{k-1}}
       b_{i_{k-1}, i_{k-2}} \dots   b_{i_2, i_{1}}\varepsilon. 
                \end{align}
              
By Proposition \ref{prop:constant}, the path-effect along \eqref{eq:path}, as defined in the dynamic path analysis (\cite{fosen}), is simply the integral of
\eqref{eq:pathEffect} in the special case $\varepsilon = 1$.  
\subsection{Direct effects}
        To isolate the effect from $X_i$ that is not mediated through the other
                covariates, we have  to attenuate any other effect the
                intervention $X_i \mapsto X_i + \epsilon$ would have had. 
        For $i \neq j$, we want to change each function $b_j$ such that 
                $X_j$ would behave as if $X_i$ was not perturbed, even if it
                was. 
        Note that 
                \begin{align}
                        b_j( X_1, \dots, X_{i-1}, X_i - \epsilon, X_{i+1},
                        \dots, X_{j-1}) =
                        b_j( X_1, \dots, X_{j-1}) - b_{j,i} \varepsilon. 
                \end{align}
        This means that if we perturbed each function $b_j$ into $b_j -
                b_{j,i} \varepsilon$ while we  also perturbed $X_i$ such that 
                $b_i$ is replaced by $b_i + \varepsilon$, 
                then $\tilde \alpha_t - \alpha_t$ represents the  direct effect.
        The sum
                \eqref{eq:mainEffect} telescopes for 
                $$V = (0, \dots,0, \varepsilon, -b_{i+1,i} \varepsilon, 
                \dots,    -b_{n,i} \varepsilon ),$$ so the direct effect equals 
                \begin{align} \label{eq:directEffect}
                        \tilde \alpha_t  -   \alpha_t = \big(\beta_t ^\intercal J_t \big)_i \varepsilon.  
                \end{align}

        Note that the direct effect, as defined in the dynamic
path analysis (\cite{fosen}), is simply the integral of
\eqref{eq:directEffect} in the special case $\varepsilon = 1$.                             
  \section*{ Proofs}
        \begin{proof} [Proof of Theorem \ref{thm:main}]
 Let $\nabla_i$ denote the $n \times
n$-dimensional matrix  where the $i$'th row equals 
$         \begin{pmatrix} 
              b_{i,1} &  \dots &  b_{i,i-1} & 0   & \dots & 0, 
        \end{pmatrix}$
and all the remaining entries  are $0$. 
                Moreover, let 
                \begin{align}
                                        \phi_s( x_1, \dots, x_n) :=  \beta_s^\intercal J_s (
                                                       I + \nabla_n)  \dots (I
                                                       + \nabla_2)
                                                       \begin{pmatrix} x_1
                                                               \\ \vdots \\ x_n
                                                       \end{pmatrix}, 
                                \end{align}
                and note that 
\begin{align*}
        \phi_t(x_1, \dots, x_n) = \phi_t( x_1 I( t \geq t_1), \dots, x_n I (
        t \geq t_n)) . 
\end{align*}

                \begin{lem} \label{OddsIndependence2}
                        Each $W_{j}$ is independent of $W_1, \dots, W_{j-1}$
                         conditionally on  $T > t
                        \geq t_j$ and \begin{align} \label{margHaz}
        \phi_t(W_1, \dots, W_{j},  E [ W_{j+1} |  T > t ]        ,
        \dots,  E [ W_{n} | T > t])
\end{align}
defines the hazard of $T$ with respect to $\mathcal F_t^j$
with respect to both $P$ and $\tilde P$. 
  \end{lem}
\begin{proof}
For every continuous function $h$
        \begin{align*}
                & E [  h(W_n) | \mathcal F_t^{n-1} \cap \{ T>t\} ] \\
                = & \frac{E [  h(W_n)  \exp( -\int_{t_n}^t \phi_s (W_1, \dots, W_n) ds   )   
                        | \mathcal F_{t_n}^{n-1} \cap \{ T>t_n\}  ]     }{E [  \exp(- \int_{t_n}^t \phi_s (W_1, \dots, W_n) ds   )   
                        | \mathcal F_{t_n}^{n-1} \cap \{ T>t_n\}  ]   } \\                        
                 = & \frac{E [  h(W_n)  \exp( -\int_{t_n}^t \phi_s (W_1, \dots,
                         W_{n-1}, 0) + \phi_s (0, \dots,
                         0, W_{n})  ds   )   
                        | \mathcal F_{t_n}^{n-1}\cap \{ T>t_n\}]     }{E [  \exp(- \int_{t_n}^t 
                        \phi_s (W_1, \dots,
                         W_{n-1}, 0) + \phi_s (0, \dots,
                         0, W_{n})  
                        ds   )   
                        | \mathcal F_{t_n}^{n-1} \cap \{ T>t_n\}  ]   } \\
                   = & \frac{E [  h(W_n)  \exp(-\int_{t_n}^t\phi_s (0, \dots,
                         0, W_{n})  ds   )   
                        | \mathcal F_{t_n}^{n-1} \cap \{ T>t_n\}  ]     }{E [ -\exp( \int_{t_n}^t 
                        \phi_s (0, \dots,
                         0, W_{n})  
                        ds   )   
                        | \mathcal F_{t_n}^{n-1} \cap \{ T>t_n\}  ]   } \\
                  = & \frac{E [  h(W_n)  \exp( -\int_{t_n}^t\phi_s (0, \dots,
                         0, W_{n})  ds   )   
                        | T > t_n]     }{E [  \exp( -\int_{t_n}^t 
                        \phi_s (0, \dots,
                         0, W_{n})  
                        ds   )   
                        | T > t_n]   }\\
                   =   & E [  h(W_n) | T>t ].
        \end{align*}
        Suppose that the independence claim is true for $ n, n-1, \dots, j+1$. 
        The innovation theorem implies that \eqref{margHaz} gives the hazard with
        respect to $\mathcal F_t^{j}$.

        Furthermore, we have that
        \begin{align*}
        & E [  h(W_j) | \mathcal F_t^{j-1} \cap \{ T>t\}] \\
                = & \frac{E [  h(W_j)  \exp( -\int_{t_j}^t \phi_s (W_1, \dots,
                        W_j, E[W_{j+1} | T >s], \dots, E[W_{n} | T >s] ) ds   )   
                        | \mathcal F_{t_j}^{j-1}\cap \{ T>t_j\}]     }{E [  \exp(- \int_{t_n}^t \phi_s (W_1, \dots,
                         W_j, E[W_{j+1} | T >s], \dots, E[W_{n} | T >s] ) ds   )   
                        | \mathcal F_{t_j}^{j-1} \cap \{ T>t_j\}]   } \\
                 = & \frac{E [  h(W_j)  \exp( -\int_{t_j}^t \phi_s (0, \dots, 0,
                         W_j, 0 , \dots, 0)  ds   )   
                        | \mathcal F_{t_j}^{j-1}\cap \{ T>t_j\}]     }{E [  \exp(- \int_{t_j}^t 
                         \phi_s (0, \dots, 0,
                         W_j, 0 , \dots, 0) 
                        ds   )   
                        | \mathcal F_{t_j}^{j-1}\cap \{ T>t_j\}]   } \\
                 = & \frac{E [  h(W_j)  \exp( -\int_{t_j}^t \phi_s (0, \dots, 0,
                         W_j, 0 , \dots, 0)  ds   )   
                        | T > t_j]     }{E [  \exp(-\int_{t_j}^t 
                         \phi_s (0, \dots, 0,
                         W_j, 0 , \dots, 0) 
                        ds   )   
                        | T > t_j]   } \\
                 = & E [  h(W_j) | T>t] 
                          \end{align*}
So the independence claim is also true for $j$. The result therefore follows by
induction. 
\end{proof}            

                \begin{lem}
                        \begin{align} \label{extCause}
                                E_{\tilde P} [h( W_i) |
                                T > t] = 
                                          E_{P} [h( W_i + V_i) | T > t]
              \end{align}
                        for every $i$ such that $t  > t_i$. 
                       \end{lem}

Note that for convenience of notation in the following proof we introduce the following notation $W_k(s):=E_{P}[W_k | T>s]$  and  $\tilde{W}_k(s):=E_{\tilde{P}}[W_k | T>s]$    
                       
                        \begin{proof}
                             We have that 
                                \begin{align*} & E_{\tilde P} [h(W_n) | T > t ]
                                        \\ 
                                         = & E_{\tilde P} [h(W_n) |\mathcal
                                         F_{t-}^{n-1} \cap \{ T>t\}]  \\ = & 
                                         \frac{  E_{\tilde P} [h(W_n)  \exp
                                                 \big( 
                                                 -\int_{t_n}^t \phi_s(W_1, \dots, W_n) ds 
                                               \big)
                                               |\mathcal
                                         F_{t_n}^{n-1}\cap \{ T>t_n\} ]     }
                                 {   E_{\tilde P} [ \exp
                                                 \big( 
                                                  -\int_{t_n}^t \phi_s(W_1, \dots, W_n) ds 
                                                                                                     \big)             |\mathcal
                                         F_{t_n}^{n-1} \cap \{ T>t_n\}  ]}\\ = & 
                                         \frac{  E_{ P} [h(W_n + V_n)  \exp
                                                 \big( -\int_{t_n}^t \phi_s(W_1, \dots, W_{n-1}, W_n+ V_n) ds 
                                               \big)             |\mathcal
                                         F_{t_n}^{n-1}  \cap \{ T>t_n\} ]    }
                                 {   E_{ P} [ \exp
                                                 \big( -\int_{t_n}^t \phi_s(W_1, \dots, W_{n-1}, W_n+ V_n) ds 
                                               \big)             |\mathcal
                                         F_{t_n}^{n-1} \cap \{ T>t_n\} ]}
                                 \\ = & 
                                         \frac{  E_{ P} [h(W_n + V_n)  \exp
                                                 \big(   -\int_{t_n}^t \phi_s(W_1, \dots, W_n) ds 
                                               \big)             |\mathcal
                                         F_{t_n}^{n-1} \cap \{ T>t_n\} ]}   
                                 {   E_{ P} [ \exp
                                                 \big(  -\int_{t_n}^t \phi_s(W_1, \dots, W_n) ds 
                                               \big)             |\mathcal
                                         F_{t_n}^{n-1} \cap \{ T>t_n\} ] } \\ = &  
                                 E_{P} [h(W_n + V_n) |\mathcal
                                         F_{t-}^{n-1} \cap \{ T>t\}   ]
                                         \\ = & E_P[ h(W_n + V_n) | T>t]
                                         .   
                                \end{align*}
                           
Suppose that \eqref{extCause} is satisfied for $i=n, n-1,\dots, j+1$. 
Now 
\begin{align*} & E_{\tilde P}[h( W_j)| T > t ] \\ = 
 & E_{\tilde P} [h(W_j) |\mathcal
                                         F_{t-}^{j-1} \cap \{T>t\} ]  \\ =                                          
                                        & 
                                         \frac{  E_{\tilde P} [h(W_j)  \exp
                                                 \big( 
                                                 -\int_{t_j}^t \phi_s(W_1, \dots, W_n) ds 
                                               \big)
                                               |\mathcal
                                         F_{t_j}^{j-1} \cap \{T>t_j\}]    }
                                 {   E_{\tilde P} [ \exp
                                                 \big( 
                                                  -\int_{t_j}^t \phi_s(W_1, \dots, W_n) ds 
                                                                                                     \big)             |\mathcal
                                         F_{t_j}^{j-1} \cap \{T>t_j\} ] }\\ = &
                                  \frac{  E_{\tilde P} [h(W_j)  \exp
                                                 \big( 
                                                 -\int_{t_j}^t \phi_s(W_1,
                                                 \dots, W_j, 
                                               \tilde{ W}_{j+1}(s), \dots,  \tilde{ W}_{n}(s)
                                                 ) ds 
                                               \big)
                                               |\mathcal F_{t_j}^{j-1}\cap \{T>t_j\}] }
                                 {   E_{\tilde P} [ \exp
                                                 \big( 
                                                  -\int_{t_j}^t \phi_s(W_1,
                                                  \dots, W_j, 
                                                  \tilde{ W}_{j+1}(s), \dots,  \tilde{ W}_{n}(s)
) ds 
                                                                                                     \big)             |\mathcal
                                         F_{t_j}^{j-1} \cap \{T>t_j\}]}\\ 
                                 = &
                                  \frac{  E_{ \tilde P} [h(W_j)  \exp
                                                 \big( 
                                                 -\int_{t_j}^t \phi_s(W_1,
                                                 \dots, W_j, 
                                                   W_{j+1}(s)+ V_{j+1}, \dots,   W_{n}(s)+ V_n
                                                 ) ds 
                                               \big)
                                               |\mathcal F_{t_j}^{j-1}\cap \{T>t_j\}]}
                                 {   E_{\tilde P} [ \exp
                                                 \big( 
                                                  -\int_{t_j}^t \phi_s(W_1,
                                                  \dots, W_j, 
                                                    W_{j+1}(s)+ V_{j+1}, \dots,   W_{n}(s)+ V_n
) ds 
                                                                                                     \big)             |\mathcal
                                         F_{t_j}^{j-1}  \cap \{T>t_j\} ]] } \\ 
                                 = &
                                  \frac{  E_{  P} [h(W_j + V_j )  \exp
                                                 \big( 
                                                 -\int_{t_j}^t \phi_s(W_1,
                                                 \dots, W_{j-1}, W_j + V_j, 
                                                 W_{j+1}(s)+ V_{j+1}, \dots,   W_{n}(s)+ V_n
                                                 ) ds 
                                               \big)
                                               |\mathcal F_{t_j}^{j-1} \cap \{T>t_j\} ] }
                                 {   E_{P} [ \exp
                                                 \big( 
                                                  -\int_{t_j}^t \phi_s(W_1,
                                                  \dots,W_{j-1} , W_j + V_j , 
                                                   W_{j+1}(s)+ V_{j+1}, \dots,   W_{n}(s)+ V_n
) ds 
                                                                                                     \big)             |\mathcal
                                         F_{t_j}^{j-1} \cap \{T>t_j\} ] }
                                 \\ =
                                  & \frac{  E_{P} [h(W_j + V_j)  \exp
                                                 \big( 
                                                 -\int_{t_j}^t \phi_s(W_1,
                                                 \dots, W_j,   W_{j+1}(s), \dots,  W_{n}(s))  + \phi_s(0,
                                                  \dots, 0, V_{j}, \dots,
                                                  V_n) ds 
                                               \big)
                                               |\mathcal
                                         F_{t_j}^{j-1}  \cap \{T>t_j\}]     }
                                 {   E_{ P} [ \exp
                                                 \big( 
                                                  -\int_{t_j}^t \phi_s(W_1,
                                                  \dots,W_{j},  W_{j+1}(s), \dots,   W_{n}(s)) + \phi_s(0,
                                                  \dots, 0, V_{j}, \dots,
                                                  V_n) ds 
                                                                                                     \big)             |\mathcal
                                         F_{t_j}^{j-1} \cap \{T>t_j\} ] }
                                 \\  = & \frac{  E_{P} [h(W_j + V_j)  \exp
                                                 \big( 
                                                 -\int_{t_j}^t \phi_s(W_1,
                                                 \dots, W_j ,  W_{j+1}(s), \dots,   W_{n}(s))) ds 
                                               \big)
                                               |\mathcal
                                         F_{t_j}^{j-1} \cap \{T>t_j\} ]]    }
                                 {   E_{ P} [ \exp
                                                 \big( 
                                                  -\int_{t_j}^t \phi_s(W_1,
                                                  \dots, W_j ,  W_{j+1}(s), \dots,   W_{n}(s) ) ds 
                                                                                                     \big)             |\mathcal
                                         F_{t_j}^{j-1}\cap \{T>t_j\} ] } \\ =
                                 & E_P [ h(W_j + V_j) |  \mathcal F_{t}^{j-1} \cap \{T>t\}  ] \\
                                  = &
                                  E_P [ h(W_j + V_j) |  T > t ],  
                         \end{align*}
                         so \eqref{extCause} follows by induction.
                        \end{proof}

 Let $\alpha_t^0$ denote the hazard of $T$ with respect to 
                       $ N_t$ and $P$, and let $\tilde \alpha_t^0$ 
                       denote the hazard we would see  if the frequencies had
                       been governed by $\tilde P$. 
                
                       Now,
                \begin{align*}
 & E_{\tilde P} [ \beta^\intercal_t J_t X | T >t] = \phi_t( E_{\tilde P}[
          W_1 |   T > t], \dots,  E_{\tilde P}[
          W_n |   T > t]
     ) \\ = &  \phi_t( E_{ P}[
          W_1 |   T > t]+ V_1, \dots,  E_{P}[
          W_n |   T > t] + V_n
     ) 
     \\  = &  \phi_t( E_{ P}[
          W_1 |   T > t], \dots,  E_{P}[
          W_n |   T > t] 
     )   +  \phi_t( V_1, \dots, V_n)
                \end{align*}
            
                By the Innovation Theorem, we have the following: 
             \begin{align} \label{effect2}
                \tilde \alpha_t^0 = \alpha_t^0 +  \beta_t ^\intercal J_t
                 ( I + \nabla_i)  ( I + \nabla_{i-1}) \dots ( I + \nabla_2)  V. 
        \end{align}
Finally, \eqref{eq:mainEffect} follows by writing out the matrix products. 
        \end{proof}

        \begin{proof}[Proof of Proposition \ref{prop:constant}] 
        \begin{align*}
              E\big[ X_{k} \big|  \mathcal F_t^{k-1}  \big]  = 
              & E \big[ \sum_{j = 1}^{k-1} b_{k,j} X_j   +   W_k \big|  \mathcal F_t^{k-1}  \big]  
                = \sum_{j = 1}^{k-1} b_{k,j} X_j   +  
            E \big[ W_k \big|  \mathcal F_t^{k-1}  \big]  
            \\ =  
              & 
             \sum_{j = 1}^{k-1} b_{k,j} X_j    +  
            E \big[ W_k \big|  T > t  \big], 
        \end{align*}
        where the last equality follows since $W_k  W_1, \dots, W_{k-1} |
        T > t$
        by Lemma \ref{OddsIndependence2}.
    \end{proof}
\end{appendix}

\newpage
\addcontentsline{toc}{section}{References}
\bibliographystyle{apalike}
\bibliography {SIM_DPA_FEB2015}{}
\end{document}